\documentclass[12pt,a4paper]{article}
\usepackage[centertags]{amsmath}
\usepackage{amsfonts,amsthm,amssymb}
\usepackage{amssymb}
\usepackage{amsmath}
\usepackage{graphicx}
\usepackage{booktabs}
\usepackage{appendix}
\usepackage[hmargin=2.54cm,vmargin=2.54cm]{geometry}
\usepackage{float}
 \usepackage{tabularray}
\usepackage{tikz}
\usetikzlibrary{patterns}
\usepackage[edges]{forest}
\usepackage{tikz}
\usepackage{soul}

\usetikzlibrary{calc}
\def\eps{0.15}
\PassOptionsToPackage{natbib=true,maxbibnames=10}{biblatex}

\usepackage[
backend=biber,
style=authoryear,
uniquename = false]{biblatex}
\renewbibmacro*{in:}{}

\newcommand{\In}[2]{#1\!\in\! #2}

\newcommand{\x}{\In{x}{X}}

\newcommand{\p}{\mathbf{P}}

\newtheorem{lemma}{Lemma}
\newtheorem{proposition}{Proposition}

\newtheorem{theorem}{Theorem}

\newtheorem{corollary}{Corollary}

\theoremstyle{definition}
\newtheorem{definition}{Definition}
\newtheorem{example}{Example}

\title{Bayesian Persuasion without Commitment}

\begin{document}
\author{Itai Arieli\thanks{Department of Economics, University of Toronto and the Faculty of Data and Decision Sciences, Technion,
	itai.arieli@utoronto.ca.} \and Colin Stewart\thanks{Department of Economics, University of Toronto, colin.stewart@utoronto.ca.}}
    \maketitle

\maketitle

\begin{abstract}
    We introduce a model of persuasion in which a sender without any commitment power privately gathers information about an unknown state of the world and then chooses what to verifiably disclose to a receiver. The receiver does not know how many experiments the sender is able to run, and may therefore be uncertain as to whether the sender disclosed all of her information. Despite this challenge, we show that, under general conditions, the sender is able to achieve the same payoff as in the full-commitment Bayesian persuasion case.
\end{abstract}

\section{Introduction}


In standard models of Bayesian persuasion, the sender chooses an information structure and commits to revealing the realized signal \citep{kamenica2011bayesian}. In many applications, the commitment power of the sender is questionable. We provide conditions under which the Bayesian persuasion solution remains an equilibrium in a model where the sender has no ability to commit.

The commitment assumption entails several logically distinct elements. First, information gathering is public: the receiver knows exactly what information structure the sender has chosen. Second, all information is fully disclosed to the receiver. Third, information is verifiable: the sender cannot fabricate or misreport signals.

We introduce a model that drops the first two of these elements while retaining the third. The sender designs and executes a number of conditionally independent experiments to learn about the underlying state. As in Bayesian persuasion, the sender has complete freedom in the design of each experiment. The number of experiments the sender can run is stochastic and unobservable to the receiver. Upon observing the outcomes of the experiments, the sender chooses which ones, if any, to disclose to the receiver. The receiver observes only the disclosed outcomes (and the experiments that led to them) and then chooses an action.

In our model, all of the sender's information acquisition is private: the receiver does not observe the number of experiments, their design, or their outcomes unless the sender chooses to disclose them. The receiver, understanding both the sender's discretion over experiment choice and the possibility of selective disclosure, updates beliefs accordingly. We identify a simple condition under which, in the sender-optimal equilibrium, the sender does as well as if she had full commitment power (that is, there is an equilibrium that implements the Bayesian persuasion solution).

As a motivating example, consider a corporation seeking to influence a regulatory decision on a major infrastructure project. To support its position, the firm can commission independent studies on various aspects such as economic impact, environmental effects, and job creation, but faces practical constraints like limited funding, time pressure, or expert availability, which are beyond its control and vary across projects. The firm strategically chooses which types of studies to conduct, observes the results, and then selectively discloses any subset of the findings to the regulator. Aware of the firm’s discretion over both the design and disclosure of studies, the regulator updates its beliefs about the project and makes a policy decision accordingly.

As another example, consider a pharmaceutical company preparing a drug approval application. The company may conduct clinical trials targeting different populations or treatment variations, but institutional constraints such as limited patient availability, recruitment delays, or regulatory approvals restrict how many trials can be executed. The company strategically selects which trials to run, observes the outcomes, and includes only a subset of trials and their results in its submission. Aware of the potential for selective disclosure, the health authority evaluates the submitted evidence with caution before making an approval decision.


Despite the absence of commitment and despite the receiver's awareness that the sender may be strategically withholding information, the sender can achieve the Bayesian persuasion benchmark in equilibrium under two simple conditions: first, the sender is always able to run at least one experiment, and second, there exists a ``punishing'' action (possibly mixed) that the sender does not prefer in any state to any of the actions chosen in the Bayesian persuasion solution. The latter condition is guaranteed to hold when the sender has state-independent preferences (also known as ``transparent motives'').

The Bayesian persuasion solution can be achieved in a relatively simple way when the sender has transparent motives. In that case, it suffices to consider a single experiment whose outcomes correspond one-to-one with actions used in the Bayesian persuasion solution. The sender simply repeats the same experiment as many times as she is able to, and always reveals only one outcome---the one associated with her most preferred action among all of the outcomes she observed. We show that there always exists an experiment for which this strategy generates the same distribution of posterior beliefs for the receiver as in Bayesian persuasion, thereby achieving the full commitment solution.

With state-dependent preferences, the sender's disclosure decisions can depend in more complicated ways on the set of outcomes she observes. Moreover, the sender's information may run counter to her interests; for example, with a binary state, it could be that the sender wants to induce a low belief in the receiver when her own belief is high, and vice versa. In this case, it suffices to consider two kinds of strategies for the sender: for any given number $k$ of experiments that she can run, she either runs $k$ identical experiments or she runs $k-1$ identical experiments along with one experiment that fully reveals the state. Either way, the sender discloses exactly one realized outcome from the repeated experiment; if she learns the state, she uses it only to inform her choice of which outcome to disclose. We employ a fixed point argument to show that there exists an experiment that generates the Bayesian persuasion solution when the sender chooses optimally from this class of strategies for each $k$. The equilibrium experiment whose outcome is disclosed can differ substantially from that in the standard one-shot optimal Bayesian persuasion policy, and yet the Bayesian persuasion posteriors obtain because the receiver accounts for the sender's strategic disclosure choice.

It is sometimes possible to obtain the Bayesian persuasion solution even if there is a positive probability that the sender cannot run any experiments. We provide a tight characterization when the sender has transparent motives. To prevent deviations by the sender, it must be that disclosing nothing leads to the worst Bayesian persuasion action, $\underline{a}$, for the sender. The Bayesian persuasion solution can be attained in equilibrium if the probability that the sender cannot run any experiments is low enough that the posterior associated with $\underline{a}$ could arise by combining the prior belief when the sender runs no experiments with the posterior obtained conditional on running experiments and disclosing the worst-case outcome.

Our main result places no restriction on the receiver's beliefs after unexpected disclosures. In Section~\ref{sec:pbe}, we introduce a refinement in which off-path beliefs must be consistent with the disclosed evidence. We show that our main result extends to this refinement under a \emph{consistent punishment condition}: for every set of states not ruled out by the evidence, the receiver has an optimal (possibly mixed) action for some belief on that set that, in each state in that set, is no better for the sender than any action recommended with positive probability in that state by the Bayesian persuasion policy. The punishment may therefore depend on the evidence. With transparent motives, this condition is equivalent to the credibility condition of \citet{dai2025bayesian}. When there is no upper bound on the number of experiments the sender may be able to run, our result also holds under their stronger equilibrium notion.

The consistent punishment condition holds, for example, in the price-discrimination environment of \citet{bergemann2015limits}. There, a consumer who privately gathers information about his valuation and selectively discloses it to a monopolist can achieve the consumer-optimal outcome without any commitment: trade is efficient, and the monopolist earns only the uniform-price monopoly profit.

The original model of Bayesian persuasion \citep{kamenica2011bayesian} features a sender who can choose any experiment and fully commits to disclosing its outcome. In an extension, they show that the sender can achieve the Bayesian persuasion payoff if the choice of experiment is public and signals are verifiable. More recently, a number of papers have relaxed the full commitment assumption in various ways. The Bayesian persuasion payoff can also be obtained (at least approximately, and under some conditions) if the sender sends cheap talk messages but has incentives to maintain a reputation \citep{best2024persuasion,mathevet2019reputation}. In other models, the sender cannot typically obtain the Bayesian persuasion payoffs. \citet{lipnowski2020cheap} analyze cheap talk environments in which the sender has no commitment power and messages are not verifiable. \citet{lipnowski2022persuasion} study intermediate cases in which, with some probability, the sender can freely change the message observed by the receiver. In \citet{koessler2023informed}, the sender commits to the experiment only after learning the state of the world. In contrast to these models, in our model, the sender discloses verifiable information and is uninformed at the time of choosing the experiment.

\citet{dai2025bayesian} independently developed a model that, like ours, features a receiver who is uncertain about how many experiments the sender can run, the outcomes of which can be verifiably disclosed. In their model, unlike ours, the sender also commits to a public experiment and acquires information sequentially. In the case of transparent motives, under a stronger notion of equilibrium than we employ for our main result, they identify conditions under which the Bayesian persuasion payoff is not attainable for the sender. \citet{felgenhauer2014strategic} and \citet{felgenhauer2017bayesian} also study disclosure by a sender who has transparent motives and acquires information sequentially. In their models, information is costly. \citet{felgenhauer2017bayesian} identify conditions on the parameters under which the Bayesian persuasion outcome is feasible with binary states and actions. \citet{lou2023private} similarly considers a sender who can run multiple sequential experiments and disclose a subset of them. Unlike the present paper and \citet{dai2025bayesian}, in his model, a sender who discloses the outcome of one experiment must also disclose the outcomes of all previous experiments. 

Beginning with \citet{glazer2004optimal}, a series of papers have analyzed games in which a sender with a private type can disclose exogenous verifiable evidence to a receiver, and have identified conditions under which the \emph{receiver} cannot benefit from commitment to a mechanism \citep{glazer2006study,sher2011credibility,hart2017evidence,ben2019mechanisms}. In contrast, we focus on the role of commitment to a disclosure policy by a \emph{sender} who is uninformed ex ante and optimally chooses what evidence to look for.


Like us, \citet{titova2025persuasion} study when disclosure of verifiable evidence can yield the Bayesian persuasion payoff for the sender. In the transparent motives case, they show that this payoff can be attained in equilibrium if and only if there exists a deterministic optimal policy under which the sender does not gain from fully revealing the state. \citet{ali2024design} study environments with state-dependent preferences in which the sender never benefits from fully revealing her type, including monopoly pricing and insurance contracting. With a rich evidence structure and an atomless distribution of types, they provide conditions under which every payoff profile attainable through information design can be approximated by an equilibrium of the disclosure game. In both papers, the sender's information is exogenous. In contrast, in our model, the sender is uninformed ex ante and privately chooses the experiments that generate the evidence available for disclosure. The sender can therefore deviate both in what information she acquires and in what she reveals. We show that, despite this additional discretion, the sender can attain the Bayesian persuasion payoff exactly under the punishing-action condition, including when her preferences are state-dependent.

\section{Model}

There is an unknown state of the world $\omega \in \Omega$, where $\Omega=\{\omega_1,\ldots,\omega_n\}$ is a finite set. The state is realized according to a common prior probability $p \in \Delta(\Omega)$. We assume without loss of generality that $p_\omega>0$ for every $\omega\in\Omega$, where $p_\omega$ denotes the probability $p$ assigns to state $\omega$. The receiver has a finite action set $A$. The receiver's utility is given by a function $u: A \times \Omega \to \mathbb{R}$ and the sender's utility is given by $\tilde{v}: A\times \Omega \to \mathbb{R}$. We assume for simplicity that for each $a\in A$, there exists a belief $q\in\Delta(\Omega)$ at which $a$ is optimal for the receiver. When the sender has state-independent preferences, we sometimes omit the $\omega$ argument from $\tilde{v}$.

The sender can acquire information about the state by choosing a collection of conditionally independent experiments and observing their realizations. Let $S$ be a countably infinite set. 
An \emph{experiment} is a mapping $\pi:\Omega\to\Delta(S)$. Together with the prior $p$, an experiment induces a probability distribution $\p_\pi\in\Delta(\Omega\times S)$. Let $\Pi$ denote the set of all experiments.

The number of experiments the sender can choose is determined at random according to some distribution $\nu \in \Delta(\mathbb{N_+})$; 
we write $\nu_k$ for the probability $\nu$ assigns to $k\in\mathbb{N}_+$ and refer to $k$ as the \emph{type} of the sender.  

The interaction between the sender and receiver takes place as follows. First, the number of experiments $k$ is drawn according to $\nu$. The realized $k$ is observed only by the sender. 
The sender then chooses $k$ experiments $\pi_1, \dots, \pi_k\in\Pi$. Thereafter, the state $\omega \in \Omega$ is realized and the experiments $\pi_1, \dots, \pi_k$ are conducted independently given $\omega$. The state is not directly observed by the sender; she observes only the outcomes $(s_1, \dots, s_k)$ of the $k$ experiments. The sender can then choose any subset $T \subseteq \{(\pi_1, s_1), \dots, (\pi_k, s_k)\}$ (with repetition allowed) to disclose to the receiver. The receiver observes only the disclosed subset $T$ and then chooses an action $a \in A$. 

Formally, a pure strategy for a sender of type $k\in\mathrm{supp}\,\nu$ is a
pair $\tau_k=(\boldsymbol{\pi}_k,d_k)$, where
$\boldsymbol{\pi}_k=(\pi_1,\ldots,\pi_k)\in\Pi^k$ specifies the experiments
to be conducted and
$d_k:S^k\to 2^{\{1,\ldots,k\}}$ specifies, for every realization of their
outcomes, the indices of the outcomes to be disclosed. A pure strategy for
the sender is a collection
$\tau=(\tau_k)_{k\in\mathrm{supp}\,\nu}$. (We also allow for mixed strategies.)

For the receiver, let $\Delta^*$ be the set of all subsets
$\{(\pi_1,s_1),\dots,(\pi_\ell,s_\ell)\}$ of $\Pi\times S$ of size $\ell$
across all finite $\ell\leq\sup\mathrm{supp}\,\nu$ (including the empty disclosure $\varnothing$, corresponding to $\ell=0$). A
belief system for the receiver is a mapping
$\rho:\Delta^*\to\Delta(\Omega)$ that assigns a posterior belief $\rho(T)$
to each $T\in\Delta^*$.\footnote{We omit beliefs about elements of the
history other than the state $\omega$ because they are not relevant for
equilibrium.} A pure strategy for the receiver is a mapping
$\sigma:\Delta^*\to A$. We allow behavioral strategies, represented by
mappings $\sigma:\Delta^*\to\Delta(A)$.


Note that the sender's strategy $\tau$ induces a joint distribution $\mathbb{P}_\tau \in \Delta(\Delta^* \times \Omega)$ over disclosures and states.

\begin{definition}
A profile $(\tau, (\rho, \sigma))$ is a \emph{weak perfect Bayesian equilibrium} (wPBE) of the game if $\tau_k$ is a best reply to $(\rho, \sigma)$ for every $k$, $\rho (T)=\mathbb{P}_\tau(\cdot |T)$ for almost every $T$ (with respect to $\tau$), and $\sigma(T)$ is optimal given the belief $\rho (T)$ for every $T\in\Delta^*$.
\end{definition}

We consider a stronger equilibrium concept in Section \ref{sec:pbe} that adds the restriction that if an off-equilibrium-path disclosure contains a realization that occurs with probability $0$ in some state $\omega$ then the receiver's belief must assign probability $0$ to $\omega$. While applying a restriction along these lines is common, in the present environment, it has bite only because the sender has complete flexibility in the design of the experiments, including experiments that lead to zero posteriors on some state(s). If instead the sender were limited to experiments leading to posteriors of at least $\varepsilon$ for each state (where $\varepsilon>0$ is small), the two equilibrium concepts would be equivalent.

One might object that because the sender is uninformed when she chooses the experiments, disclosure of an off-equilibrium-path experiment should affect the receiver's belief about the state only through the information contained in the experiment. Note, however, that such a disclosure could affect the receiver's belief about $k$, which could in turn affect how the receiver reasons about what information the sender may have acquired and chosen to hide.

\section{Main Result}

Note that, in any wPBE, the distribution of posterior beliefs for the receiver must be Bayes plausible. Thus the maximal expected utility that the sender can achieve in any wPBE is bounded above by the Bayesian persuasion utility. We ask whether the sender can achieve this utility in our setting despite the absence of commitment power.

\begin{definition}
Given a subset of actions $B\subseteq A$, the receiver has a \emph{punishing action} with respect to $B$ if there exists a posterior belief $q\in\Delta(\Omega)$ and a distribution over actions $x\in\Delta(A)$ such that $x$ is a best reply for the receiver to the belief $q$ and $\sum_{a\in A}x(a)\tilde{v}(a,\omega)\leq \tilde{v}(b,\omega)$ for every $b\in B$ and $\omega\in\Omega$.     
\end{definition}

Thus a punishing action exists with respect to $B$ if the action distribution $x$ is worse for the sender than any $b\in B$ at every state $\omega\in\Omega$. For example, if the sender has transparent motives, then for every $B$, there exists a punishing action with respect to $B$ (namely, the sender's least preferred action in $B$).

We henceforth assume for simplicity that there exists a unique Bayesian persuasion policy. By the revelation principle, we can identify this policy with a direct recommendation policy $P\in\Delta(A\times\Omega)$ (whose marginal over $\Omega$ is equal to the prior, $p$). Let $\tilde A$ be the set of actions that are played with positive probability under $P$.

\begin{theorem}\label{th:information}
If the receiver has a punishing action with respect to $\tilde A$, then there exists a wPBE in which the sender achieves the Bayesian persuasion utility. 
\end{theorem}

Proofs are in the Appendix.


In the transparent motives case, one can construct an equilibrium in which the sender achieves the Bayesian persuasion utility; see Appendix \ref{appendix:thm2}. In the general case, such a construction appears to be intractable as the sender's strategy can depend qualitatively on the distribution and realization of sender types. We therefore employ a fixed point argument to prove the existence of such an equilibrium.

We first illustrate the theorem with a simple example.

\begin{example}\label{ex:product_adoption} Consider the classic product adoption example where the state space $\Omega=\{\omega_0,\omega_1\}$ represents the quality of a product which can be low ($\omega=\omega_0$) or high ($\omega=\omega_1$), with equal prior probability. The action set is $A=\{a_0,a_1\}$. Action $a_0$ represents a non-adoption action that yields a payoff of $0$ to both sender and receiver. Action $a_1$ represents adoption and yields a payoff of $1$ to the sender but is risky for the receiver, with $u(a_1,\omega_0)=-3$ and $u(a_1,\omega_1)=1$. Thus, the receiver would like to adopt the product if and only if she assigns probability at least $3/4$ to state $\omega_1$. The sender's indirect utility is depicted in Figure \ref{fig:product_adoption}.

\begin{figure}[h]
\centering
\begin{tikzpicture}[x=3.5cm, y=3.0cm]

 \draw[->] (0,0) -- (0,1.2) node[above left] {Sender's utility};
  \draw[->] (0,0) -- (1.1,0) node[below right] {Receiver's belief};

  \foreach \x in {0, 0.5, 0.75, 1}
    \draw (\x,0) -- (\x,0.03);

  \node[below] at (0,-0.03) {$0$};
  \node[below] at (0.5,-0.03) {$\tfrac{1}{2}$};
  \node[below=6pt] at (0.75,0) {$\tfrac{3}{4}$};
  \node[below] at (1,-0.03) {$1$};

  \draw (-0.02,1) -- (0.02,1) node[left=4pt] {$1$};

  \draw[blue, very thick] (0,0) -- (0.75,0);   
  \draw[blue, very thick] (0.75,1) -- (1,1);   

  \filldraw[blue, fill=white, thick] (0.75,0) circle (0.016);
  \filldraw[blue] (0.75,1) circle (0.016);

  \draw[dashed] (0,0) -- (0.75,1);

\end{tikzpicture}
\caption{The sender's indirect utility for Example \ref{ex:product_adoption}. The dashed line depicts the concavification associated with the Bayesian persuasion solution.}\label{fig:product_adoption}
\end{figure}
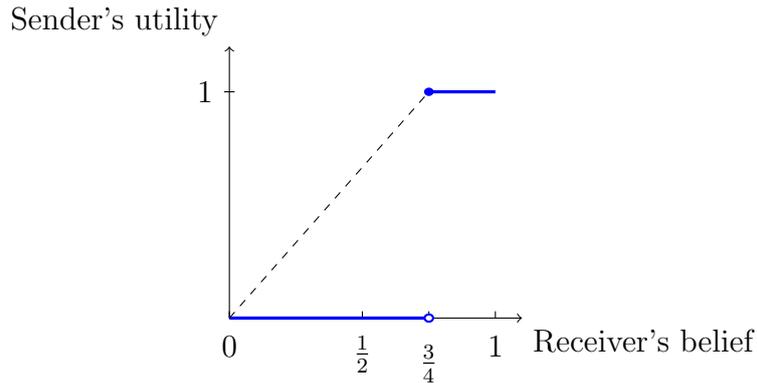

The Bayesian persuasion utility is implemented by the distribution of posteriors $\frac{1}{3}\delta_0+\frac{2}{3}\delta_{\frac 3 4}$, where $\delta_q$ denotes the distribution placing unit mass on $q$. 
Suppose the number of experiments available to the sender is distributed geometrically with parameter $r$, that is, that $\nu_k=(1-r)^{k-1}r$ for each $k$. We will illustrate how to generate a wPBE that implements the Bayesian persuasion distribution.

Consider the experiment $F$ that generates the posterior distribution
$$(1-\alpha)\delta_0+\alpha\delta_{\frac{r+2}{2r+2}},$$
where $\alpha=\frac{2r+2}{2r+4}.$ This posterior distribution is achieved by an experiment $F:\Omega\to \Delta(S)$ for $S=\{s_0,s_1\}$ with $F(\omega_1)[s_1]=1$ and $F(\omega_0)[s_1]=\frac{r}{r+2}.$

Consider the following strategy $\tau$ for the sender. For each $k$, repeat the experiment $F$ $k$ times. If there is at least one realized $s_1$ signal, reveal $\{(F,s_1)\}$ to the receiver; otherwise, reveal $\{(F,s_0)\}$ to the receiver.


We first claim that, on the equilibrium path, the receiver's beliefs in a wPBE must be $\rho(\{(F,s_1)\})=3/4$ if $\{(F,s_1)\}$ is revealed and $\rho(\{(F,s_0)\})=0$ if $\{(F,s_0)\}$ is revealed. The latter follows immediately from the construction of $F$; for the former, we must show that 
$$P_\tau(\omega_1|\{(F,s_1)\})=\frac{3}{4}.$$
To see this, note that the signal $s_1$ is realized with probability one given $\omega_1$. Therefore, $P_\tau(\{(F,s_1)\}|\omega_1)=1.$
By definition,
\begin{align*}
    P_\tau(\{(F,s_0)\}|\omega_0)=\sum_{k=1}^\infty r(1-r)^{k-1}\left(1-\frac{r}{r+2}\right)^k=\frac{2}{3}.
\end{align*}
Therefore, $P_\tau(\omega_1|\{(F,s_1)\})=3/4$, as desired. 
It is straightforward to verify that the sender has no profitable deviation.

To interpret the construction, notice that the positive signal in the experiment has a posterior $\frac{r+2}{2r+2}\geq\frac{3}{4}$, with the inequality being strict when $r<1$. Thus, in order to compensate for the possibility that the sender can run multiple experiments when $r<1$, to persuade the receiver to take action $a_1$, a high realization from any one experiment must provide stronger evidence in favor of the high state than does the high signal in the Bayesian persuasion solution.  
\end{example}

\subsection{Outline of Proof of Theorem \ref{th:information}}

It turns out that the equilibrium construction in Example \ref{ex:product_adoption} applies more generally as long as the sender has transparent motives (in which case a punishing action always exists). In that case, as we show in the appendix, there exists an experiment $F:\Omega\to\Delta(S)$ where each signal $s_j\in S=\{s_1,\dots,s_m\}$ corresponds to an action $a_j$ such that $\tilde v(a_j)\leq \tilde v(a_i)$ for $j<i$. The strategy for the sender of type $k$ is to draw $k$ conditionally i.i.d.\ signals $\{s^1,\ldots,s^k\}$ according to $F$ and then disclose $\{(F,s_i)\}$ where $i=\max\{j: s_j=s^l \text{ for some } l\}$. The receiver then plays the corresponding action $a_i$.

Outside of the transparent motives case, the sender's preferences over actions may depend on her posterior belief, making it no longer possible to order actions by the sender's preference. The sender's optimal disclosure strategy may therefore depend in complicated ways on the realized profile of signals. To overcome this difficulty, we use a fixed point argument. The following example illustrates the basic logic.

\begin{example}\label{ex:state-dependent}
Let $\Omega=\{\omega_0,\omega_1\}$. The prior is $1/2$, where we identify the prior and posterior with the probability of $\omega_1$. The receiver's action set is $A=\{a_0,a_1,a_2\}$. Action $a_0$ is optimal for the receiver for posteriors in $[0,1/3]$, action $a_1$ is optimal for posteriors in $[1/3,2/3]$, and action $a_2$ is optimal for posteriors in $[2/3,1]$. The sender's utility is $\tilde v(a_0,\omega_0)=0$, $\tilde v(a_0,\omega_1)=3/2$, $\tilde v(a_1,\omega)=0$ for each $\omega$, $\tilde v(a_2,\omega_1)=0$, and $\tilde v(a_2,\omega_0)=3/2$. Note that $a_1$ is a punishing action. The sender's indirect utility is depicted in Figure \ref{fig:state-dependent}.

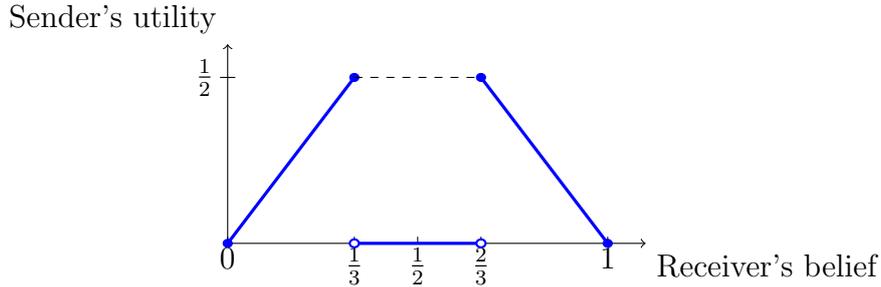
\begin{figure}[h]
\centering
\begin{tikzpicture}[x=5cm, y=4.4cm]
\draw[->] (0,0) -- (0,0.6) node[above left] {Sender's utility};
  \draw[->] (0,0) -- (1.1,0) node[below right] {Receiver's belief};

  \foreach \x/\lab in {0/{0}, 0.333333/{\tfrac{1}{3}}, 0.5/{\tfrac{1}{2}}, 0.666667/{\tfrac{2}{3}}, 1/{1}}
    \draw (\x,0) -- (\x,0.02) node[below] {$\lab$};

  \draw (-0.02,0.5) -- (0.02,0.5) node[left=4pt] {$\tfrac{1}{2}$};

  \draw[blue, very thick] (0,0) -- (0.333333,0.5);          
  \draw[blue, very thick] (0.333333,0) -- (0.666667,0);      
  \draw[blue, very thick] (0.666667,0.5) -- (1,0);          

  \filldraw[blue, fill=white, thick] (0.333333,0) circle (0.012);
  \filldraw[blue, fill=white, thick] (0.666667,0) circle (0.012);

  \filldraw[blue] (0,0) circle (0.012);
  \filldraw[blue] (0.333333,0.5) circle (0.012);
  \filldraw[blue] (0.666667,0.5) circle (0.012);
  \filldraw[blue] (1,0) circle (0.012);

  \draw[dashed] (0.333333,0.5) -- (0.666667,0.5);

\end{tikzpicture}
\caption{The sender's indirect utility for Example \ref{ex:state-dependent}.}\label{fig:state-dependent}
\end{figure}

Notice that the sender's interests are opposed to the receiver's in the sense that, if they had the same belief, they would never agree on the optimal action: whenever $a_2$ is optimal for the receiver the sender prefers $a_0$, and vice versa. Thus, for example, a sender with belief above $1/2$ would ideally like to induce a belief of at most $1/3$ in the receiver. One might expect this conflict to create difficulties in persuading the receiver, yet Theorem \ref{th:information} indicates that the Bayesian persuasion payoff is achievable for the sender.

Let $\nu=(\nu_k)_{k\geq 1}$ be any distribution of types for the sender. Consider any experiment $F:\Omega\to\Delta\left(\{s_1,s_2\}\right)$. Let $x=F(\omega_1)[s_1]$ and $y=F(\omega_0)[s_1]$.

Suppose the receiver plays action $a_0$ if the sender discloses $\{(F,s_1)\}$, plays action $a_2$ if the sender discloses $\{(F,s_2)\}$, and plays action $a_1$ otherwise. It follows that whenever the sender assigns posterior probability greater than $1/2$ to $\omega_1$, she will disclose $\{(F,s_1)\}$ whenever she has received at least one signal equal to $s_1$; similarly, whenever she assigns probability less than $1/2$ to $\omega_1$ she will disclose $\{(F,s_2)\}$ whenever it is feasible to do so. This disclosure strategy generates a probability distribution $P\in\Delta(\Omega\times A)$ where, for example, $P(a_0|\omega_1)$ is the conditional probability that action $a_0$ is played in state $\omega_1$.

Given this disclosure strategy, define a mapping by
$$T(x,y)=\left(x+\alpha^*\left(\frac{1}{3}-P(a_0|\omega_1)\right),y+\alpha^*\left(\frac{2}{3}-P(a_0|\omega_0)\right)\right)$$
for $\alpha^*> 0$. If this mapping has a fixed point, it means that the corresponding experiment together with the above strategies lead to $P(a_0|\omega_1)=1/3$ and $P(a_0|\omega_0)=2/3$, as in the Bayesian persuasion solution. It is straightforward to verify that this forms a wPBE.

A couple of modifications are needed to ensure that $T(x,y)$ has a fixed point. First, we need the mapping to be continuous. By allowing the sender to choose any mixture of optimal disclosures when she is indifferent, we obtain a correspondence satisfying the required continuity property. Second, in some cases, at least some types of the sender may prefer to deviate and run one experiment that fully reveals the state; the result of this fully revealing experiment is never disclosed to the receiver, but can increase the sender's utility because it enables her to tailor her disclosure choice according to the state. To accommodate this possibility, we simply allow for the option to run one fully revealing experiment as part of the sender's optimal strategy.
Finally, $\alpha^*>0$ can be chosen sufficiently small that $T$ maps
$[0,1]^2$ into itself.
\end{example}

We now outline the general argument. Write
$\tilde A=\{a_1,\ldots,a_m\}$ and represent the Bayesian persuasion policy by
the matrix $Q^*$, where $Q^*_{ij}=P(a_j\mid\omega_i)$. For a candidate
row-stochastic matrix $Q$, let $\pi_Q$ be the experiment with signals
$S=\{s_1,\ldots,s_m\}$ satisfying
$\pi_Q(\omega_i)[s_j]=Q_{ij}$.

Consider fixing the receiver's strategy to be the one that, for each $j$, chooses action $a_j$ following a disclosure of $\{(\pi_Q,s_j)\}$ and chooses the punishing action after any other disclosure. For this strategy of the receiver, suppose we show that there is an optimal strategy for the sender that induces $Q^*$ using disclosures of the form $\{(\pi_Q,s_j)\}$. Because the conditional probability of disclosing $s_j$ in state $\omega_i$ is $Q^*_{ij}=P(a_j\mid\omega_i)$, Bayes' rule implies that the receiver's posterior after $s_j$ is exactly the Bayesian persuasion posterior $P(\cdot\mid a_j)$. Hence, her on-path action is optimal and the induced joint distribution over states and actions is the Bayesian persuasion policy $P$. Thus it suffices to show that such an optimal strategy for the sender exists.


For each type $k$, define an auxiliary decision problem $\Gamma_k(Q)$. The
sender either conducts $k$ independent copies of $\pi_Q$, or, when $k>1$,
conducts $k-1$ copies of $\pi_Q$ and uses the remaining experiment to learn
the state fully. After observing her information, she selects one realized
signal $s_j$ from $\pi_Q$ to disclose and receives payoff $\tilde v(a_j,\omega)$. Thus,
the receiver is not a strategic player in this auxiliary problem: because
disclosure of $s_j$ induces action $a_j$, the sender's choice of which signal
to disclose is equivalent to choosing the receiver's action from among those
made available by her realized signals. Let $C_{\Gamma_\nu}(Q)$ be the set
of matrices induced by optimal strategies in these decision problems,
aggregated across sender types.

These decision problems are sufficient to capture all relevant deviations in the original game.
Disclosure of an outcome from an experiment other than $\pi_Q$ induces the
punishing action and is therefore unprofitable. If the outcome of a different
experiment is concealed, that experiment is useful only as a source of
private information for choosing which signal from $\pi_Q$ to disclose. It
can therefore be replaced by a fully revealing experiment, which is already
available in $\Gamma_k(Q)$.

It remains to find an experiment $Q$ for which an optimal sender strategy
induces $Q^*$. This does not generally hold when $Q=Q^*$, because selective
disclosure changes the distribution of observed signals. For $\alpha>0$ and row-stochastic matrices $Q$ satisfying $Q_{ij}=0$ whenever $Q^*_{ij}=0$, define
\[
T_\alpha(Q)
=
\{Q+\alpha(Q^*-Q'):Q'\in C_{\Gamma_\nu}(Q)\}.
\]
For sufficiently small $\alpha$, this is a nonempty, convex- and
compact-valued correspondence from the relevant set of row-stochastic
matrices into itself and has a closed graph. Kakutani's Theorem therefore
gives a fixed point. At a fixed point,
$Q=Q+\alpha(Q^*-Q')$ for some $Q'\in C_{\Gamma_\nu}(Q)$, and hence
$Q'=Q^*$, as needed.


\section{Examples}

In this section, we analyze two examples. The first shows how the equilibrium experiment disclosed to the receiver can be qualitatively different from the one used in the Bayesian persuasion solution. The second shows that the conclusion of Theorem \ref{th:information} can fail if there is no punishing action. 

\begin{example}\label{ex:weird_example}
Let $\Omega=\{\omega_0,\omega_1\}$ with prior $p=1/2$. Suppose the sender's type is equal to $3$ with probability one. There are $4$ actions $A=\{a_1,a_2,a_3,a_4\}$. The sender's utility is $\tilde v (a_1,\omega)=\tilde v (a_4,\omega)=-1/2$ for each $\omega$, $\tilde v (a_2,\omega_0)=\tilde v (a_3,\omega_1)=1/2$, and $\tilde v (a_2,\omega_1)=\tilde v (a_3,\omega_0)=-1/2$. Actions $a_1$, $a_2$, $a_3$, and $a_4$ are optimal for the receiver on the intervals $[0,1/3]$, $[1/3,1/2]$, $[1/2,2/3]$, and $[2/3,1]$, respectively. The sender's indirect utility is depicted in Figure \ref{fig:weird_example}.

\begin{figure}[h]
\centering
\begin{tikzpicture}[x=4cm, y=4cm]

  \draw[->] (0,-0.6) -- (0,0.55) node[above left, yshift=4pt] {Sender's utility};
  \draw[->] (0,0) -- (1.1,0) node[below right] {Receiver's belief};

  \draw (0,0) -- (0,0.02);
  \node[below left] at (0,0) {$0$};
  \foreach \x/\lab in {0.333333/{\tfrac{1}{3}}, 0.5/{\tfrac{1}{2}}, 0.666667/{\tfrac{2}{3}}, 1/{1}}
    \draw (\x,0) -- (\x,0.02) node[below] {$\lab$};

  \draw (-0.02,0.5) -- (0.02,0.5) node[left=4pt] {$\tfrac{1}{2}$};
  \draw (-0.02,0.166667) -- (0.02,0.166667) node[left=4pt] {$\tfrac{1}{6}$};
  \draw (-0.02,-0.5) -- (0.02,-0.5) node[left=4pt] {$-\tfrac{1}{2}$};

  \draw[blue, very thick] (0,-0.5) -- (0.333333,-0.5);
  \filldraw[blue] (0,-0.5) circle (0.016);
  \filldraw[blue, fill=white, thick] (0.333333,-0.5) circle (0.016);

  \draw[blue, very thick] (0.333333,0.166667) -- (0.5,0) -- (0.666667,0.166667);
  \filldraw[blue] (0.333333,0.166667) circle (0.016);
  \filldraw[blue] (0.666667,0.166667) circle (0.016);

  \draw[blue, very thick] (0.666667,-0.5) -- (1,-0.5);
  \filldraw[blue, fill=white, thick] (0.666667,-0.5) circle (0.016);
  \filldraw[blue] (1,-0.5) circle (0.016);

  \draw[dashed] (0.333333,0.166667) -- (0.666667,0.166667); 
  \draw[dashed] (0.333333,0.166667) -- (0,-0.5);            
  \draw[dashed] (0.666667,0.166667) -- (1,-0.5);            

\end{tikzpicture}
\caption{The sender's indirect utility for Example \ref{ex:weird_example}. The dashed line depicts the concavification. The distribution of posteriors that implements the Bayesian persuasion policy is $\frac{1}{2}\delta_{\frac{1}{3}}+\frac{1}{2}\delta_{\frac{2}{3}}$.}\label{fig:weird_example}
\end{figure}
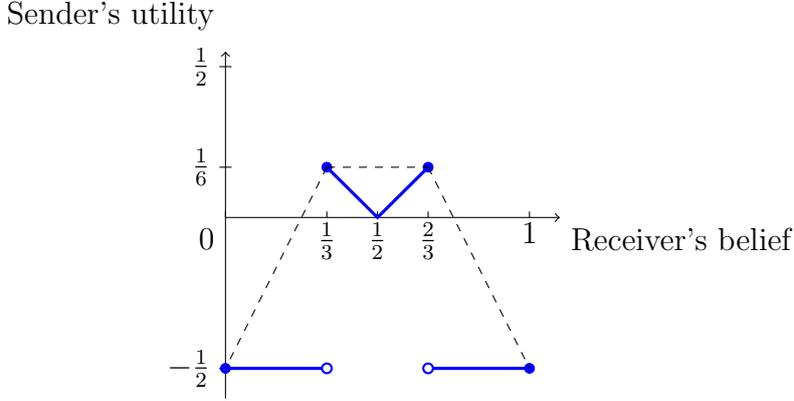
 
We first try to construct an equilibrium of the following form. The sender runs $3$ conditionally i.i.d.\ symmetric binary experiments $F:\Omega\to \Delta\left(\{s_2,s_3\}\right)$ and discloses $\{(F,s_i)\}$ if and only if signal $s_i$ occurs a majority of the time (i.e., either $2$ or $3$ times). After observing disclosure $\{(F,s_i)\}$, the receiver plays action $a_i$. These strategies result in the Bayesian persuasion solution if $F(\omega_1)[s_3]=F(\omega_0)[s_2]=q$ where $q\approx0.613$ is the solution of $q^3+{3 \choose 2}(q^2(1-q))=2/3$. For such an experiment, the sender's disclosure strategy is to reveal the majority signal and the receiver's strategy is a best reply to the sender's since $P_F(\omega_1 \mid \text{majority of } s_2)=1/3$ and $P_F(\omega_1 \mid \text{majority of } s_3)=2/3$.

We claim, however, that the sender has a profitable deviation. Consider instead using one of the three experiments to learn the state while running the other two experiments as above according to $F$. The sender then discloses $\{(F,s_3)\}$ in state $\omega_1$ if she has received at least one $s_3$ signal, and discloses $\{(F,s_2)\}$ in state $\omega_0$ if she has received at least one $s_2$ signal. This is indeed a profitable deviation since the conditional probability of playing action $a_3$ in state $\omega_1$ is at least $3/4$ (and similarly for $a_2$ in state $\omega_0$), which exceeds the corresponding probability of $2/3$ in the above strategy profile.

Given that the sender learns the state, in equilibrium, the receiver accounts for the likelihood that disclosure of $\{(F,s_3)\}$ means the sender learned the state is $\omega_1$. This effect tends to increase the receiver's belief following disclosure of $\{(F,s_3)\}$. How, then, can the sender's strategy ensure that this belief is not greater than $2/3$? Perhaps counterintuitively, the resolution involves reversing the direction of the signals so that $s_3$ is relatively \emph{more} likely in state $\omega_0$ than in state $\omega_1$. 

For the exact construction, take $F(\omega_1)[s_2]=F(\omega_0)[s_3]=\frac{1}{\sqrt{3}}\approx0.577$. The sender's strategy is as follows. She runs two experiments according to $F$ and one that fully reveals the state. She then discloses $\{(F,s_3)\}$ in state $\omega_1$ whenever doing so is feasible, and discloses $\{(F,s_2)\}$ in state $\omega_0$ whenever feasible. The receiver then plays action $a_3$ if and only if he observes signal $s_3$. Thus the receiver plays $a_3$ despite the fact that the disclosed signal, if taken at face value, suggests that the state is more likely to be $\omega_0$. 

To see that this is indeed a wPBE, we need to verify first that the posterior distribution generated by the sender's strategy is the desired one. Indeed, the probability of revealing signal $s_3$ at state $\omega_1$ is $1-(1/\sqrt{3})^2=2/3$, as desired, and similarly for $s_2$. The sender's disclosure strategy, given her information, maximizes her utility by construction. The only deviation that needs to be considered is that the sender could draw $3$ signals according to $F$ instead of learning the state and reveal the minority signal when both signals occur (and the only available signal otherwise). Under this deviation, the probability of revealing signal $s_3$ in state $\omega_1$ is
$$\left(1-\frac{1}{\sqrt{3}}\right)^3+3\left(1-\frac{1}{\sqrt{3}}\right)\left(\frac{1}{\sqrt{3}}\right)^2<\frac{2}{3}.$$
Therefore, the sender has no profitable deviation.
\end{example}

\begin{example}\label{ex:no_punishing_action}
This example shows that if there is no punishing action, the conclusion of Theorem \ref{th:information} may not hold.

Let $\Omega=\{\omega_0,\omega_1\}$ and $A=\{a_0,a_1\}$. The prior, which we identify with the probability of state $\omega_1$, is $1/2$. Action $a_0$ is optimal for the receiver if and only if he assigns posterior probability at most $2/3$ to state $\omega_1$ and action $a_1$ is optimal if and only if he assigns  probability at least $2/3$ to state $\omega_1$. The utility of the sender is $\tilde v(a_0,\omega)=1$ for each $\omega\in\Omega$,  $\tilde v(a_1,\omega_1)=2/3 +\varepsilon$, and $\tilde v(a_1,\omega_0)=5/3 +\varepsilon$ for some $\varepsilon>0$. The indirect utility of the sender is depicted in Figure \ref{fig:no_punishing_action}. As is evident from the graph, the Bayesian persuasion policy induces a posterior distribution of $\frac{1}{4}\delta_0+\frac{3}{4}\delta_{\frac{2}{3}}$, giving the sender an expected payoff of $1+3\varepsilon/4$.

\begin{figure}[h]
\centering
\begin{tikzpicture}[x=4.5cm, y=3.5cm]

  \draw[->] (0,0) -- (0,{1+\eps+0.1}) node[above left] {Sender's utility};
  \draw[->] (0,0) -- (1.1,0) node[below right] {Receiver's belief};

  \draw (1,0) -- (1,0.03);
  \draw (0.666667,0) -- (0.666667,0.03);
  \draw (0.5,0) -- (0.5,0.03);

  \node[below] at (1,-0.03) {$1$};
  \node[below] at (0.666667,-0.03) {$\tfrac{2}{3}$};
  \node[below] at (0.5,-0.03) {$\tfrac{1}{2}$};

  \draw (-0.02,1) -- (0.02,1) node[left=4pt] {$1$};
  \draw (-0.02,{1+\eps}) -- (0.02,{1+\eps}) node[left=4pt] {$1+\varepsilon$};

  \draw[blue, very thick] (0,1) -- (0.666667,1);

  \filldraw[blue, fill=white, thick] (0.666667,1) circle (0.016);

  \draw[blue, very thick]
    (0.666667,{1+\eps}) -- (1,{0.666667+\eps});

  \filldraw[blue] (0.666667,{1+\eps}) circle (0.016);
  \filldraw[blue] (1,{0.666667+\eps}) circle (0.016);

  \draw[dashed] (0.666667,{1+\eps}) -- (0,1);

\end{tikzpicture}
\caption{The sender's indirect utility for Example \ref{ex:no_punishing_action}. The dashed line depicts the concavification.}
\label{fig:no_punishing_action}
\end{figure}

Suppose the distribution of sender types $k$ is given by $\nu_1=\nu_2=\frac{1}{2}$. We claim that, for $\varepsilon<1/12$, the Bayesian persuasion policy cannot be obtained in any equilibrium. 

To see why, suppose there was such an equilibrium. Then it must be that whenever the receiver chooses action $a_0$, his belief is equal to $0$. It follows that, in equilibrium, the sender's belief is also always equal to $0$ when the receiver chooses action $a_0$. Moreover, the receiver chooses action $a_1$ whenever the state is $\omega_1$. 

Now consider the action played by the receiver if the sender discloses nothing. Since a sender with belief $0$ strictly prefers action $a_1$ to action $a_0$, and $a_0$ is played with positive probability when the sender has belief $0$, it must be that the receiver plays $a_0$ with probability one following no disclosure. In particular, when $k=1$, the sender's expected payoff is at least $1$.

Now consider type $k=2$. One (not necessarily optimal) option for this type is to use one experiment to replicate the experiment of type $1$ and use the other experiment to learn the state, and then to follow type $1$'s disclosure strategy in state $\omega_0$ and disclose nothing in state $\omega_1$. Relative to type $1$, this strategy gives the same expected payoff in state $\omega_0$ and in state $\omega_1$ gives a payoff of $1$ instead of $2/3+\varepsilon$. Therefore, writing $u_k$ for the expected payoff of type $k$, we have
\[
u_2\geq u_1 +\frac{1}{2}\left( 1 - \left(\frac{2}{3}+\varepsilon\right)\right).
\]
It follows that
\begin{align*}
\frac{u_1+u_2}{2}&\geq u_1 + \frac{1}{4}\left(\frac{1}{3}-\varepsilon\right)\\
&\geq \frac{13}{12}-\frac{\varepsilon}{4}\\
&>1+\frac{3\varepsilon}{4},
\end{align*}
where the second line follows from the observation that $u_1\geq 1$ and the last line follows since $\varepsilon<1/12$. This contradicts that $1+3\varepsilon/4$ is the sender's expected payoff in the Bayesian persuasion solution.
\end{example}

\section{Extensions}

\subsection{Robustness to ignorance}\label{sec:robustness}

We now consider whether our main result is robust to the possibility that the sender is ``ignorant'', i.e., that she may not be able to run any experiments. We first consider the case of transparent motives---which ensures the existence of a punishing action---and show that our main result extends when the posterior associated with the sender’s least-preferred on-path action is interior and the probability of ignorance is sufficiently small. As part of the proof, we show in the appendix that with transparent motives, one can prove Theorem \ref{th:information} in a more constructive way without relying on a fixed point argument.



The distribution of sender types is now \(\nu\in\Delta(\mathbb N)\), allowing a positive probability \(\nu_0\) that the sender cannot run any experiments. Let \(P\in\Delta(A\times\Omega)\) denote the unique Bayesian persuasion policy, and let \(\tilde A=\{a_1,\ldots,a_m\}\) be the actions used with positive probability under \(P\). Assume that \(a_1\) is the sender’s uniquely least-preferred action in \(\tilde A\):
\[
\tilde v(a_1)<\tilde v(a_j)
\qquad\text{for every }j>1.
\]

\begin{theorem}\label{thm:robustness}
 If the sender has transparent motives, there exists a wPBE in which the sender achieves the Bayesian persuasion utility if and only if
\[
\nu_0\leq\min_{\omega\in\Omega}P(a_1\mid\omega).
\]
\end{theorem}

The condition has a simple interpretation. When the sender cannot run any experiments, she necessarily discloses nothing. To prevent a sender who would otherwise induce \(a_1\) from deviating to silence, silence must also induce \(a_1\). Since the probability of running no experiments is \(\nu_0\) in every state, implementing the Bayesian persuasion policy requires that this policy assign probability at least \(\nu_0\) to \(a_1\) in every state. The theorem shows that this necessary condition is also sufficient.

In particular, the Bayesian persuasion outcome remains attainable for all sufficiently small positive values of \(\nu_0\) if and only if \(P(a_1\mid\omega)>0\) for every state \(\omega\). Equivalently, the posterior \(q^1=P(\cdot\mid a_1)\) associated with the sender’s least-preferred on-path action must lie in the interior of \(\Delta(\Omega)\).

The conclusion of Theorem~\ref{thm:robustness} also holds with
state-dependent sender preferences, provided that an on-path action
$a_1\in\tilde A$ satisfies
$\tilde v(a_1,\omega)<\tilde v(a_j,\omega)$ for every
$\omega\in\Omega$ and every $a_j\in\tilde A\setminus\{a_1\}$.
The sender's ranking of the remaining actions may vary across states.


\begin{example}\label{ex:robustness}
If the sender does not have transparent motives, interiority of all Bayesian persuasion posteriors is not sufficient to ensure that the Bayesian persuasion payoff can be attained when $\nu_0>0$ is sufficiently small. Consider again the environment from Example \ref{ex:state-dependent}. Let $\nu$ be any distribution of $k$ that contains $0$ in its support. Suppose for contradiction that there exists a wPBE implementing the Bayesian persuasion policy. Since the evidence $\varnothing$ is on the equilibrium path, the receiver's belief after observing $\varnothing$ must be either $1/3$ or $2/3$, and he must correspondingly play either action $a_0$ or $a_2$ with probability one. Assume without loss of generality that he plays action $a_0$. Note that whenever the sender's posterior belief lies above $1/2$, she strictly prefers  action $a_0$ to action $a_2$. Therefore, it must be that the probability of $\omega_1$ conditional on $a_2$ is at most $1/2$, for otherwise the sender would have a profitable deviation to disclosing $\varnothing$ instead of every disclosure that leads to $a_2$. This contradicts that the probability of $\omega_1$ conditional on $a_2$ in the Bayesian persuasion solution is $2/3$.
\end{example}

\subsection{From weak PBE to PBE}
\label{sec:pbe}

Our analysis has so far employed weak perfect Bayesian equilibrium. Under this
notion, the receiver's beliefs are determined by Bayes' rule on the equilibrium
path but are unrestricted following disclosures that occur with probability
zero. We now return to the case of $\nu_0=0$ and strengthen the equilibrium notion by requiring off-path beliefs to respect the information contained in the disclosed evidence.

For any feasible disclosure $T\in\Delta^*$, define its \emph{consistency set}
by
\[
S(T)=\left\{\omega\in\Omega:\pi(\omega)[s]>0
\text{ for every }(\pi,s)\in T\right\},
\]
where $S(\varnothing)=\Omega$. Thus, $S(T)$ consists of the states in which all
the evidence contained in $T$ could have been generated.

\begin{definition}\label{def:pbe}
A profile $(\tau,(\rho,\sigma))$ is a \emph{perfect Bayesian equilibrium}
(PBE) if it is a wPBE and
$\operatorname{supp}(\rho(T))\subseteq S(T)$ following every feasible
disclosure $T\in\Delta^*$.
\end{definition}

Thus, under PBE, evidence that rules out a state must be taken at face value,
even when it is disclosed off the equilibrium path. For simplicity, we use
PBE to refer to this refinement of wPBE. Our notion is slightly weaker than
the perfect extended Bayesian equilibrium notion used by
\citet{dai2025bayesian}, which imposes additional consistency requirements on
off-path beliefs. As discussed below, when the sender's capacity has unbounded
support, our result also holds under their stronger notion.

Recall that $P\in\Delta(A\times\Omega)$ denotes the unique Bayesian persuasion
policy, that $\tilde A=\{a_1,\ldots,a_m\}$ is the set of actions used by $P$,
and that $p^j=P(\cdot\mid a_j)$ is the posterior associated with action $a_j$.

\begin{definition}\label{def:consistent-punishment}
The Bayesian persuasion policy $P$ satisfies the \emph{consistent punishment
condition} if, for every nonempty $S\subseteq\Omega$, there exist a belief
$q_S\in\Delta(\Omega)$, with $\operatorname{supp}(q_S)\subseteq S$, and an
action distribution $x_S\in\Delta(A)$ such that $x_S$ is a best reply for the
receiver to $q_S$ and
\[
\sum_{a\in A}x_S(a)\tilde v(a,\omega)
\leq \tilde v(a_j,\omega)
\]
for every $j$ such that $p^j(S)>0$ and every
$\omega\in S\cap\operatorname{supp}(p^j)$.
\end{definition}

The condition requires a receiver-optimal punishment for each possible
consistency set. Unlike the punishing action used in
Theorem~\ref{th:information}, the punishment may depend on the disclosed
evidence, but it must be optimal under a belief supported on the states that
the evidence does not rule out.

Under state-independent sender preferences, and the uniqueness assumption
maintained by \citet{dai2025bayesian}, the consistent punishment condition is
equivalent to their requirement that every posterior induced by the Bayesian
persuasion policy be credible. Under their monotonicity assumption, this
condition characterizes when they guarantee implementation of the Bayesian
persuasion outcome for every distribution of the sender's capacity.

\begin{proposition}\label{prop:pbe}
Suppose that the Bayesian persuasion policy satisfies the consistent
punishment condition. Then there exists a PBE in which the sender achieves the
Bayesian persuasion utility.
\end{proposition}

\citet{dai2025bayesian} employ a stronger equilibrium concept than that of Proposition \ref{prop:pbe}.  When $\nu$ has unbounded support, the conclusion of Proposition~\ref{prop:pbe} continues to hold under their stronger notion. Following any finite off-path disclosure, the receiver can regard it as possible that the sender had sufficient capacity to conduct and conceal an additional experiment. For each
consistency set $S$, this additional experiment can be chosen so that one of
its concealed outcomes induces the belief $q_S$. Adapting the skeptical-belief
construction in the proof of Proposition~2 of \citet{dai2025bayesian} then
shows that the off-path beliefs used above satisfy their stronger consistency
requirements.

Implementation is more involved in our model than in that of
\citet{dai2025bayesian}. In their model, the sender can conduct a publicly
observed initial experiment and therefore use the Bayesian persuasion
experiment itself on the equilibrium path. In our model, all experimentation
is private, and selective disclosure changes the distribution of information
observed by the receiver. We must therefore use the fixed-point construction
from the proof of Theorem~\ref{th:information} to find a private experiment
and disclosure strategy that reproduce the Bayesian persuasion outcome.

\begin{example}
Consider the price-discrimination environment of \cite{bergemann2015limits}. 
The state is a consumer's valuation
$\omega\in\Omega=\{\omega_1,\ldots,\omega_n\}$, where
$0<\omega_1<\cdots<\omega_n$. The sender is the consumer, who seeks to
maximize his surplus, and the receiver is a monopolist who chooses
a price $a\in A=\Omega$. With production costs normalized to zero and
the consumer purchasing whenever his valuation weakly exceeds the
price, the receiver's payoff is
$u(a,\omega)=a\mathbf{1}_{\{\omega\geq a\}}$, and the sender's payoff
is $\tilde v(a,\omega)=(\omega-a)_+$.

In the Bayesian persuasion benchmark, the consumer can commit to a
policy for revealing information about his valuation before learning
it. The optimal policy therefore solves the consumer-surplus
maximization problem studied by \citet{bergemann2015limits}.
In our model, the consumer instead privately acquires information
about his valuation and chooses what evidence to disclose, without
commitment.

The consistent punishment condition holds for every Bayesian
persuasion policy in this environment. To see this, fix a nonempty
consistency set $S\subseteq\Omega$, and let the monopolist assign
probability one to $\bar\omega_S=\max S$. The price
$a_S=\bar\omega_S$ is optimal under this belief and gives the consumer
zero surplus at every state in $S$. Since consumer surplus is
nonnegative at every price, we have
$\tilde v(a_S,\omega)=0\leq\tilde v(a_j,\omega)$ for every
$\omega\in S$ and every on-path price $a_j$. Thus, following any
off-path disclosure, the monopolist can consistently believe that
the consumer has the highest valuation compatible with the evidence
and charge that valuation.


Consequently, Proposition~\ref{prop:pbe} yields a PBE attaining the consumer-optimal payoff: trade is efficient, and the monopolist earns only the uniform-price monopoly profit.
\end{example}

\subsection{Sequential information acquisition}
Theorem~\ref{th:information} and Proposition~\ref{prop:pbe} both extend to a setting in which the sender chooses
her experiments sequentially and observes each outcome before selecting the
next experiment. As in the simultaneous model, an experiment other than the
baseline experiment constructed in the proof cannot profitably be disclosed,
since doing so induces the punishing action. Such an experiment can therefore
benefit the sender only through the private information it provides and can be
replaced by a fully informative experiment. Moreover, the fully informative
experiment can be postponed until the final stage. Once the state is known,
any remaining experiments that are useful for disclosure can be taken to be
copies of the baseline experiment, and their outcomes are conditionally
independent of the order in which the experiments are conducted.

The auxiliary decision problem \(\Gamma_k(Q)\) must therefore be modified only
slightly. For \(k>1\), the sender first observes the outcomes of \(k-1\) copies
of the baseline experiment. As a function of these outcomes, she then chooses
whether to use her final opportunity to conduct one additional copy of the
baseline experiment or to learn the state. Her subsequent disclosure decision
may depend on all the information she has obtained. For \(k=1\), the decision
problem is unchanged. The modified problem still has a finite strategy set,
and the compactness, continuity, and probability bounds used in the fixed-point
argument continue to hold. The remainder of the proof therefore carries over
without change.
\appendix
\section{Proof of Theorem \ref{th:information}}
Let $\tilde A=\{a_1,\ldots,a_m\}$ be the set of actions played with positive probability, where $m>1$ (for $m=1$, the result is trivial). The Bayesian persuasion problem admits an optimal policy supported on at most
$n$ posteriors. Since the optimal policy is unique, it follows that $m\leq n$. For any state $\omega_i\in\Omega$ and action $a_j\in\tilde A$, let $$Q^*_{ij}=P(a_j|\omega_i)$$  
be the conditional probability of a recommendation of action $a_j$ given state $\omega_i$.

Consider a single-agent decision problem $\Gamma_k(Q)$ for the sender. Let $M_{n\times m}$ denote the set of $n\times m$ matrices with real-valued entries. The decision problem is defined by a row-stochastic matrix $Q=\{q_{ij}\}_{1\leq i\leq n, 1\leq j\leq m}\in M_{n\times m}$, 
a natural number $k\geq 1$, and a set $S=\{s_1,\ldots,s_m\}$ of $m$ signals---one for each action $a_j\in\tilde A$. 

In the first stage, if $k>1$, the sender chooses between two levels of information acquisition: obtaining $k$ signals or $k-1$ signals. If $k=1$, the sender has no choice and must select $k=1$. After this choice, the state $\omega_i \in \Omega$ is realized according to the prior $p$. Conditional on the realized state, the signal-generating matrix $Q$ determines the distribution of signals, where $P_Q(s_j \mid \omega_i) = q_{ij}$.

If the sender chooses to acquire $k$ signals, then $k$ conditionally i.i.d.\ draws $s^1, \ldots, s^k$ are generated, with each signal $s_j$ occurring with probability $q_{ij}$ in state $\omega_i$. If the sender chooses $k-1$, then $k-1$ conditionally i.i.d.\  signal draws are generated in the same way; in this case, however, the sender also observes the realized state $\omega_i$ in addition to the $k-1$ signals.

In the second stage, after observing her information, the sender selects one of the signals she has received and reveals it. Her utility depends on the revealed signal: if she reveals $s^l$ such that $s^l = s_j$, her payoff is $\tilde v(a_j, \omega_i)$.

Formally, for $k>1$, a pure strategy $\psi \in \Psi_k$ for the sender in problem $\Gamma_k(Q)$ is an element of $\{k-1,k\} \times \Lambda_k$, where the first component corresponds to the sender’s choice of information acquisition, and $\lambda \in \Lambda_k$ specifies which signal to reveal based on the sender’s information. In particular, when the sender has $k$ signals $(s^1, \ldots, s^k)$, $\lambda$ maps this vector to a revealed signal $s \in \{s^1, \ldots, s^k\}$. When the sender has $k-1$ signals $(s^1, \ldots, s^{k-1})$ and additionally observes the state $\omega_i \in \Omega$, $\lambda$ maps $(s^1, \ldots, s^{k-1}, \omega_i)$ to a signal $s \in \{s^1, \ldots, s^{k-1}\}$. For $k=1$, the sender has no choice, as she possesses only one signal, which she must reveal.

We note that any strategy $\psi$ induces a probability distribution $P^\psi$ over $\Omega \times S$, where $P^\psi(s_j | \omega_i)$ denotes the conditional probability that the sender reveals signal $s_j$ in state $\omega_i$. This distribution gives rise to a row-stochastic matrix 
$Q^\psi = \{q_{ij}^\psi\}_{1 \leq i \leq n,\, 1 \leq j \leq m}$, where 
$q_{ij}^\psi = P^\psi(s_j \mid \omega_i)$.

We next establish the following lemma.

\begin{lemma}\label{lemma:convexity}
For every $Q$, let $\tilde\Psi_k(Q)$ be the set of mixed strategies in
$\Gamma_k(Q)$ that maximize the sender's utility. Then
\[
C_{\Gamma_k}(Q)
=
\{Q^{\tilde\psi}:\tilde\psi\in\tilde\Psi_k(Q)\}
\]
is convex and compact. Moreover, the correspondence
$Q\mapsto C_{\Gamma_k}(Q)$ has a closed graph.
\end{lemma}

\begin{proof}
The set of pure strategies in $\Gamma_k(Q)$ is finite. Hence,
$\tilde\Psi_k(Q)$ is convex and compact, and its image
$C_{\Gamma_k}(Q)$ under the linear mapping
$\tilde\psi\mapsto Q^{\tilde\psi}$ is also convex and compact.
The closed-graph property follows from the finiteness of the strategy set and
the continuity of payoffs and induced outcome matrices in $Q$. Indeed, if
$Q_n\to Q$, $R_n\to R$, and $R_n\in C_{\Gamma_k}(Q_n)$, compactness of the
mixed-strategy simplex allows us to take a convergent subsequence of the
corresponding optimal strategies. Its limit is optimal in $\Gamma_k(Q)$ and
induces $R$. Hence, $R\in C_{\Gamma_k}(Q)$.
\end{proof}

Now consider a compound decision problem $\Gamma_\nu(Q)$. In this problem, the value $k$ is drawn according to the distribution $\nu$ and the sender continues as in $\Gamma_k(Q)$. Let $K$ be the support of $\nu$. A pure strategy $\psi=(\psi_k)_{k\in K}$ in $\Gamma_\nu(Q)$ is just a concatenation of pure strategies $\psi_k$ in $\Gamma_k(Q)$ for every $k\in K$. 
Note that a strategy $\psi=(\psi_k)_{k\in K}$ in $\Gamma_\nu(Q)$ is utility maximizing if and only if each $\psi_k$ is utility maximizing in $\Gamma_k(Q)$. Therefore, extending the definition of $C_{\Gamma_k}(Q)$ to $C_{\Gamma_\nu}(Q)$ in the obvious way, we have
$$C_{\Gamma_\nu}(Q)=\sum_{k}\nu_kC_{\Gamma_k}(Q).$$
In particular, $C_{\Gamma_\nu}(Q)$ is convex and compact.

Let $\mathcal{Q}$ be the set of all row-stochastic matrices $Q\in M_{n\times m}$ such that $Q_{ij}=0$ whenever $Q^*_{ij}=0$. Note that $\mathcal{Q}$ is convex and compact.

For $\alpha>0$, define a correspondence $T_\alpha:\mathcal{Q}\to M_{n\times m}$ by
$$T_\alpha(Q)=\{Q+\alpha(Q^*-Q')\mid Q'\in C_{\Gamma_\nu}(Q)\}.$$
Note that, by Lemma \ref{lemma:convexity}, $T_\alpha$ is a convex and compact-valued correspondence.

\begin{lemma}\label{lem:range-Q}
 There exists $\alpha^*>0$ for which the range of $T_{\alpha^*}$ is contained in $\mathcal{Q}$.   
\end{lemma}
\begin{proof}
For $q\in[0,1]$, let $\Pi(q)=\mathbb{E}_\nu\left[1-(1-q)^k\right]$; when $q=Q(s|\omega)$, since $1-(1-q)^k$ is the probability of getting at least one $s$ signal out of $k$ signals, $\Pi(q)$ is the maximal conditional probability of revealing a signal $s$ in state $\omega$. Note that $\Pi$ is continuous, strictly increasing, and satisfies $\Pi(0)=0$ and $\Pi(1)=1$.

Let
$a_{\min}=\min\{Q^*_{ij}: Q^*_{ij}>0\}\in(0,1],$
and
$q_0=\Pi^{-1}(a_{\min})\in(0,1].$ If $a_{\min}=1$, let $\alpha^*=1$. Otherwise, let
$\alpha^*=\min\{q_0/(1-a_{\min}),1\}$. In either case,
$\alpha^*\in(0,1]$ and
$\alpha^*(1-a_{\min})\leq q_0$. 

Fixing any $Q\in\mathcal{Q}$ and any $Q'\in C_{\Gamma_\nu}(Q)$, let 
$$\tilde Q=Q+\alpha^*(Q^*-Q').$$
Consider a given $(i,j)$. Write
$$
q=Q_{ij},\quad a=Q^*_{ij},\quad q'=Q'_{ij},\quad\text{and}\quad \tilde q = q+\alpha^*(a-q').$$

We claim that $\tilde q$ is nonnegative. If $q=0$, then signal $s_j$
never appears in state $\omega_i$, and hence $q'=0$. Therefore,
$\tilde q=\alpha^*a\geq0$. Moreover, if $a=0$, then the definition of
$\mathcal{Q}$ implies that $q=0$, and hence $q'=0$ and $\tilde q=0$.

Now suppose $q\in (0,q_0]$. Recall that $q'$ represents the probability of revealing signal $s_j$ in state $\omega_i$. Since $\Pi(q)$ is the maximal probability that $s_j$ can be revealed given $q$, we have $q'\le \Pi(q)\le \Pi(q_0)=a_{\min}\le a$, which implies $a-q'\ge 0$ and thus $\tilde q\ge q\ge 0$.

If $q>q_0$, then
\[
\tilde q
= q-\alpha^*(q'-a)\ge q-\alpha^*(1-a_{\min}) 
\ge q_0-\alpha^*(1-a_{\min})\ge q_0-q_0=0,
\]
where the final inequality follows from the definition of $\alpha^*$. This completes the proof of the claim that $\tilde q$ is nonnegative.

The fact that $\sum_{j}\tilde Q_{ij}=1$ for every $i$ is immediate from the definition of $\tilde Q$ and the corresponding property of $Q$, $Q^*$, and $Q'$.

We have shown that for every $(i,j)$, $\tilde q_{ij}\ge 0$, the row sums of $\tilde Q$ equal one, and $\tilde Q_{ij}=0$ whenever $Q^*_{ij}=0$. Therefore, $\tilde Q\in\mathcal{Q}$.
\end{proof}

Next, we show that $T_{\alpha^*}$ has a fixed point.

\begin{lemma}
There exist $\alpha^*>0$ and $Q\in\mathcal Q$ such that
$Q\in T_{\alpha^*}(Q)$.
\end{lemma}

\begin{proof}
Let $\alpha^*$ be as in Lemma~\ref{lem:range-Q}. The correspondence
$T_{\alpha^*}$ is nonempty, convex-valued, and compact-valued, and its range
is contained in $\mathcal Q$. It remains to verify that it has a closed graph.

Let $Q_\ell\to Q$ and $Z_\ell\to Z$, with
$Z_\ell\in T_{\alpha^*}(Q_\ell)$. For each $\ell$, choose
$R_\ell\in C_{\Gamma_\nu}(Q_\ell)$ such that
$Z_\ell=Q_\ell+\alpha^*(Q^*-R_\ell)$. For every
$k\in\operatorname{supp}(\nu)$, choose
$R_\ell^k\in C_{\Gamma_k}(Q_\ell)$ such that
$R_\ell=\sum_k\nu_kR_\ell^k$.

By a diagonal-subsequence argument, we may assume that
$R_\ell^k\to R_*^k$ for every $k$. Lemma~\ref{lemma:convexity} implies that
$R_*^k\in C_{\Gamma_k}(Q)$. Since all entries of these matrices are bounded
by one, bounded convergence gives
\[
R_\ell=\sum_k\nu_kR_\ell^k
\longrightarrow
\sum_k\nu_kR_*^k=:R_*.
\]
Thus, $R_*\in C_{\Gamma_\nu}(Q)$. Taking limits in the equation defining
$Z_\ell$ gives $Z=Q+\alpha^*(Q^*-R_*)$, and hence
$Z\in T_{\alpha^*}(Q)$. Therefore, $T_{\alpha^*}$ has a closed graph, and the
result follows from the Kakutani fixed-point theorem.
\end{proof}
As a corollary of the existence of a fixed point, we have
\begin{corollary}
There exists a matrix $Q\in\mathcal{Q}$ such that $Q^*\in C_{\Gamma_{\nu}}(Q)$.
\end{corollary}
\begin{proof}
 Let $Q$ be the fixed point of $T_{\alpha^*}$. By definition there exists $Q'\in C_{\Gamma_{\nu}}(Q)$ such that
 $$Q= Q+\alpha^*(Q^*-Q').$$
 The corollary follows since $\alpha^*>0$. 
  \end{proof}
 
 \begin{proof}[\textbf{Proof of Theorem \ref{th:information}}]
Let $Q\in\mathcal{Q}$ be such that $Q^*\in C_{\Gamma_{\nu}}(Q)$.
Note that any pure strategy $\psi=( \psi_k)_{k\in K}$ in $\Gamma_\nu(Q)$ induces a strategy for the sender in our evidence game as follows. Let $S=\{s_1,\ldots,s_m\}$ and let $\pi:\Omega\to\Delta(S)$ be the
experiment that induces $Q$; that is,
$\pi(\omega_i)[s_j]=Q_{ij}$ for every $1\leq i\leq n$ and
$1\leq j\leq m$. The strategy induced by $\psi$ is defined naturally as follows: for every realized $k$, if $\psi_k$ chooses to acquire $k$ signals, then the sender acquires $k$ experiments $\pi_1,\ldots,\pi_k$, all equal to $\pi$. After the $k$ signals $s^1,\ldots,s^k$ are realized, the sender reveals to the receiver a unique realization $(\pi,s)$ where the signal $s\in \{s^1,\ldots,s^k\}$ is chosen according to the corresponding $\lambda$. 
If $\psi_k$ acquires $k-1$ signals, then the sender chooses $k-1$ experiments $\pi_1,\ldots,\pi_{k-1}$, all equal to $\pi$, and uses the extra experiment to choose $\tilde \pi$, which is the full information experiment. That is, $\tilde S=\{t_1,\ldots,t_n\}$ and $\tilde \pi(\omega_i)[t_j]=\delta_{ij}$. After the $k-1$ signals $s^1,\ldots,s^{k-1}$ are realized and the state $\omega_i$ is revealed by signal $t_i$, the sender chooses to reveal the signal $s\in \{s^1,\ldots,s^{k-1}\}$ that is induced by $\lambda$ with respect to $(s^1,\ldots,s^{k-1},\omega_i)$. Note that this construction readily extends to mixed strategies $\tilde\psi$. Since $Q^*\in C_{\Gamma_{\nu}}(Q)$ by construction, there exists an optimal strategy $\tilde \psi=(\tilde \psi_k)_{k\in K}$ in $\Gamma_\nu(Q)$ such that $Q^{\tilde\psi}=Q^*$. 

Let the sender's strategy be the one that is induced by $\tilde\psi$. Define an assessment $(\rho,\sigma)$ for the receiver as follows.\footnote{Recall that an assessment $(\rho,\sigma)$ comprises a belief system $\rho$ together with a strategy $\sigma$.} Let $\{p^j\}_{j=1}^m\subset\Delta(\Omega)$ be the posteriors induced by the Bayesian persuasion policy, where $p^j=P(\cdot \mid a_j)$ is the conditional probability over $\Omega$ given that action $a_j$ is recommended by the sender. 
Define the receiver's assessment following a disclosure $\{(\pi,s)\}$ such that $s=s_j$ to be $\rho(\{(\pi,s_j)\})=p^j$ and $\sigma(\{(\pi,s_j)\})=\delta_{a_j}$. That is, if the receiver observes the experiment $\pi$ with a corresponding signal $s_j$, then his posterior belief is $p^j$ and he plays action $a_j$ with probability one. For any other $T\in\Delta^*$ let $\rho(T)=q$ and $\sigma(T)=x$, where $x\in\Delta(A)$ is the punishing action with respect to $\tilde A$, and $q\in\Delta(\Omega)$ is a posterior belief for which $x$ is optimal for the receiver.

We claim that the above profile forms a wPBE. First consider the receiver. Let $\tilde P\in\Delta(\Omega\times S)$ be the probability distribution induced by the sender's strategy; namely, $\tilde P(\omega_i,s_j)$ is the probability that the realized state is $\omega_i$ and the revealed signal is $s_j$. 
By construction, $\tilde P(s_j\mid \omega_i)=Q^*_{ij}$. Since the marginal of $\tilde P$ on $\Omega$ is the prior $p$, it must be that $\tilde P$ induces the same probability distribution as the Bayesian persuasion policy $P\in\Delta(\Omega\times A)$ in the sense that $\tilde P(\omega_i,s_j)=P(\omega_i,a_j)$. Therefore, the belief described by $\rho$ is the correct belief on the equilibrium path. Moreover, $\sigma$ is optimal given $\rho$ by construction.

We claim that the sender's strategy is optimal given the receiver's strategy. It is enough to show that it is optimal conditional on every realized $k\in K$. Note that, by construction, the sender's strategy is optimal subject to the constraint that the only available experiments to her are the experiment $\pi$ defined above and the full information experiment $\tilde \pi$. By the construction of the receiver's strategy and the definition of a punishing action, revealing evidence that differs from $\{(\pi,s)\}$ is never a profitable deviation. Thus, the only possible deviation that needs to be considered is that one or more of the $k$ experiments is conducted according to $\pi'\neq\pi$ but is revealed to the receiver with probability zero. Therefore, the sender only uses the other experiment to acquire further information about the state. Replacing any such experiments with $\tilde\pi$ makes the sender better informed and thus can only improve her utility.  Therefore, the sender's strategy maximizes her expected utility given the receiver's strategy.
\end{proof}

\section{Proof of Theorem \ref{thm:robustness}}\label{appendix:thm2}

In order to prove Theorem \ref{thm:robustness}, we allow the number of samples $k$ to depend on the realized state as well. Thus, in this case, we have $\nu\in\Delta(\Omega\times\mathbb{N})$; we write $\nu_\omega^k$ for $\nu(k\mid\omega)$. Given an experiment $\pi\in\Pi$, let $\mathbb{P}_{\pi,\nu}\in\Delta\left(\Omega\times\left(\cup_{k=0}^\infty S^k \right)\right)$ be the distribution that is generated by $\pi$ and $\nu$ when $k$ is chosen according to $\nu$ given $\omega$ and thereafter $k$ conditionally i.i.d.\ signals are drawn according to $\pi$.

\begin{proposition}\label{prop}
 Let \(\sum_{j=1}^{l}\gamma_j\delta_{q^j}\in\Delta(\Delta(\Omega))\) be a belief distribution with \(\gamma_j>0\) for every \(j\) and barycenter \(p\). For any distribution \(\nu\in\Delta(\Omega\times\mathbb N_+)\) whose marginal on \(\Omega\) is \(p\), there exist a signal set \(S=\{s_1,\ldots,s_l\}\) and an experiment \(\pi\) such that, for every \(j=1,\ldots,l\),
$
\mathbb P_{\pi,\nu}\!\left(
\max\{i:s^m=s_i\text{ for some }m\in\{1,\ldots,k\}\}=j
\right)=\gamma_j$
and\\
$
\mathbb P_{\pi,\nu}\!\left(
\,\cdot\;\middle|\;
\max\{i:s^m=s_i\text{ for some }m\in\{1,\ldots,k\}\}=j
\right)=q^j.$
\end{proposition}

\begin{proof}
For each state \(\omega\), define
\[
G_\omega(t)=\sum_{k=1}^{\infty}\nu_\omega^k t^k,
\qquad t\in[0,1].
\]
Since the number of experiments is positive, this function is continuous and strictly increasing from zero to one. It therefore has an inverse \(G_\omega^{-1}\).
We want to construct an experiment with signal set \(S=\{s_1,\ldots,s_l\}\) such that revealing the highest realized signal generates the target belief distribution with weights \(\gamma_1,\ldots,\gamma_l\) and posteriors \(q^1,\ldots,q^l\).
Define the experiment by
\[
\pi(s_j\mid\omega)
=
G_\omega^{-1}\!\left(
\frac{\sum_{i=1}^{j}\gamma_iq^i_\omega}{p_\omega}
\right)
-
G_\omega^{-1}\!\left(
\frac{\sum_{i=1}^{j-1}\gamma_iq^i_\omega}{p_\omega}
\right),
\]
where an empty sum equals zero.
This is a valid experiment. The cumulative sums inside the inverse are nondecreasing, so all signal probabilities are nonnegative. Moreover, Bayes plausibility gives
\[
\sum_{i=1}^{l}\gamma_iq^i_\omega=p_\omega.
\]
The signal probabilities therefore telescope to \(G_\omega^{-1}(1)-G_\omega^{-1}(0)=1\).
We next verify the induced distribution. For any \(j\), the probability that a single draw has index at most \(j\) is
\[
\sum_{i=1}^{j}\pi(s_i\mid\omega)
=
G_\omega^{-1}\!\left(
\frac{\sum_{i=1}^{j}\gamma_iq^i_\omega}{p_\omega}
\right).
\]
Conditional on \(\omega\), the probability that the highest realized signal index is at most \(j\) is
\[
\begin{aligned}
&\sum_{k=1}^{\infty}\nu_\omega^k
\left[
G_\omega^{-1}\!\left(
\frac{\sum_{i=1}^{j}\gamma_iq^i_\omega}{p_\omega}
\right)
\right]^k\\
&\qquad =
G_\omega\!\left(
G_\omega^{-1}\!\left(
\frac{\sum_{i=1}^{j}\gamma_iq^i_\omega}{p_\omega}
\right)
\right)
=
\frac{\sum_{i=1}^{j}\gamma_iq^i_\omega}{p_\omega},
\end{aligned}
\]
where the first expression follows from conditional independence and the equality follows from the definition of \(G_\omega\). Taking differences between the expressions for \(j\) and \(j-1\), and multiplying by \(p_\omega\), gives
$
\mathbb{P}_{\pi,\nu}
\left(
\omega,\,
\max\{i:s^m=s_i\text{ for some }m\in\{1,\ldots,k\}\}=j
\right)
=
\gamma_jq^j_\omega.$
Summing over \(\omega\) shows that the highest realized signal is \(s_j\) with probability \(\gamma_j\). Bayes’ rule then implies that the conditional distribution of the state given this event is \(q^j\), as required.
\end{proof}

We next provide an alternative proof of
Theorem \ref{th:information} for the transparent motives case where the sender's utility is state independent.

\begin{proof}[\textbf{Alternative Proof of Theorem \ref{th:information} under Transparent Motives}]
Let $\alpha_1\delta_{q^1}+\cdots+\alpha_l \delta_{q^l}\in\Delta(\Delta(\Omega))$ be the optimal policy of the sender. Without loss of generality, we can assume that each posterior $q^i$ corresponds to a different optimal action $a^i$ of the receiver, and that $\tilde v(a^i)\leq\tilde v(a^j)$ for $i<j$. Let $\pi$ be the experiment that implements the distribution according to Proposition \ref{prop} with $S=\{s_1,\ldots,s_l\}$. Consider the strategy for the sender that, for each type $k$, runs the experiment $\pi$ $k$ times and reveals to the receiver the highest of the realized signals she observes, that is, she reveals $\{(\pi,s_j)\}$, where $j=\max \{i \mid s^m=s_i \text{ for some }m\in\{1,\dots,k\}\}$.

The assessment $(\rho,\sigma)$ of the receiver is as follows. On the equilibrium path, the belief is given by $\rho(\{(\pi,s_j)\})=q^j$ for each $j$ and the strategy is $\sigma(\{(\pi,s_j)\})=a^j$. Let $a_*\in A$ be the worst action for the sender that is rationalizable for the receiver and let $q_*\in\Delta(\Omega)$ be a belief for which $a_*$ is optimal. For any $T\in\Delta^*$ that is off the equilibrium path, let $\rho(T)=q_*$ and $\sigma(T)=a_*$. It follows by construction that this profile forms a weak perfect Bayesian equilibrium.
\end{proof}


\begin{proof}[\textbf{Proof of Theorem \ref{thm:robustness}}]
We return to the original model, in which \(k\) is independent of \(\omega\). Order the on-path actions so that \(\tilde v(a_1)<\tilde v(a_2)\leq\cdots\leq\tilde v(a_m)\), and write \(q^j=P(\cdot\mid a_j)\).
Suppose first that \(\nu_0\leq\min_{\omega\in\Omega}P(a_1\mid\omega)\). If \(\nu_0=1\), the condition implies that \(P\) always recommends \(a_1\). The receiver can then hold belief \(p\) and choose \(a_1\) following every disclosure, yielding the desired equilibrium. Henceforth, suppose that \(\nu_0<1\).
Define a joint distribution \(P'\) on \(\tilde A\times\Omega\) by
\[
P'(a_j,\omega)
=
\frac{P(a_j,\omega)-\nu_0p_\omega\mathbf 1_{\{j=1\}}}{1-\nu_0}.
\]
The assumed inequality ensures that all entries are nonnegative, and summing over \(j\) gives \(\sum_jP'(a_j,\omega)=p_\omega\). Thus, \(P'\) is a probability distribution with state marginal \(p\).
Let \(\nu'_k=\nu_k/(1-\nu_0)\) for \(k\geq1\). Apply Proposition \ref{prop} to the posterior distribution associated with \(P'\) and the capacity distribution \(\nu'\). This yields an experiment \(\pi\) for which the joint probability of state \(\omega\) and highest realized signal \(s_j\) is \(P'(a_j,\omega)\). Any zero-probability action is omitted, retaining the order of the remaining signal indices.
For each \(k\geq1\), the sender conducts \(k\) independent copies of \(\pi\). If the highest realized signal is \(s_j\) with \(j>1\), she discloses \(\{(\pi,s_j)\}\); otherwise, she remains silent. Type \(k=0\) also remains silent. Following \(\{(\pi,s_j)\}\) for \(j>1\), the receiver holds belief \(q^j\) and chooses \(a_j\). Following any other disclosure, she holds belief \(q^1\) and chooses \(a_1\).
For \(j>1\), the joint probability of state \(\omega\) and disclosure \(\{(\pi,s_j)\}\) is \((1-\nu_0)P'(a_j,\omega)=P(a_j,\omega)\). The joint probability of state \(\omega\) and silence is \(\nu_0p_\omega+(1-\nu_0)P'(a_1,\omega)=P(a_1,\omega)\). Hence, the receiver’s on-path beliefs satisfy Bayes’ rule, and the induced distribution of states and actions is \(P\). Her actions are optimal at every specified belief.
The sender optimally discloses the signal inducing her most-preferred available action. Changing her experiments cannot improve her payoff: disclosing an outcome from any experiment other than \(\pi\) induces \(a_1\), and additional private information does not affect her ranking of actions. Replacing such an experiment with another copy of \(\pi\) can only expand her available choices. The proposed profile is therefore a wPBE.
Conversely, suppose that a wPBE achieves the Bayesian persuasion utility. By uniqueness, its induced distribution of states and actions is \(P\). The required inequality is immediate if \(\nu_0=0\), so suppose that \(\nu_0>0\).
Any on-path receiver response assigning positive probability to \(a_1\) must choose \(a_1\) with probability one. Otherwise, replacing that response by a better action in its support would preserve receiver optimality and increase the sender’s payoff above the Bayesian persuasion optimum. Since \(P(a_1)>0\), some on-path disclosures therefore yield the sender exactly \(\tilde v(a_1)\).
Silence occurs on path, so its response is supported on \(\tilde A\). It must induce \(a_1\) with probability one; otherwise, a sender whose disclosure induces \(a_1\) would strictly prefer silence. Since type \(k=0\) necessarily remains silent and \(k\) is independent of \(\omega\), we obtain
\[
P(a_1\mid\omega)\geq\Pr(k=0\mid\omega)=\nu_0
\qquad\text{for every }\omega\in\Omega.
\]
Taking the minimum over states proves necessity.
\end{proof}

\section{Proof of Proposition \ref{prop:pbe}}

\begin{proof}
Consider the matrix $Q$, the experiment $\pi$, and the sender's strategy
constructed in the proof of Theorem~\ref{th:information}. Under this strategy,
the sender discloses a single outcome
$T_j=\{(\pi,s_j)\}$, and the resulting joint distribution over the state and
the receiver's action coincides with $P$. Following $T_j$, let the receiver
hold belief $p^j$ and choose action $a_j$. Following any other feasible
disclosure $T$, let the receiver hold belief $q_{S(T)}$ and choose
$x_{S(T)}$. The receiver's actions are optimal at the specified beliefs, and
the resulting belief system satisfies
$\operatorname{supp}(\rho(T))\subseteq S(T)$ following every disclosure.

It remains to verify the sender's incentives. Consider a private history at
which the prescribed disclosure is $T_j$. Every state that is possible at
this history belongs to $\operatorname{supp}(p^j)$. If the sender instead
makes an off-path disclosure $T$, each such state also belongs to $S(T)$.
The consistent punishment condition therefore implies, state by state, that
the sender's payoff from $x_{S(T)}$ is no greater than her payoff from $a_j$.
Thus, no noncanonical disclosure is profitable.

The argument concerning deviations in information acquisition is the same as
in the proof of Theorem~\ref{th:information}. An experiment different from
$\pi$ cannot profitably be used to produce evidence for disclosure, since
disclosing its outcome induces a consistent punishment. Such an experiment
can therefore benefit the sender only through the private information it
provides. For $k>1$, any collection of such experiments is weakly less
informative than a single fully informative experiment, after which the
remaining opportunities can be used to conduct copies of $\pi$. The resulting
strategy is one of the alternatives available in the decision problem
$\Gamma_k(Q)$. For $k=1$, an experiment different from $\pi$ can lead only to
a noncanonical disclosure and is directly dominated by using $\pi$.
Optimality in $\Gamma_\nu(Q)$ therefore rules out every deviation in
information acquisition. The resulting strategy profile is a PBE and
implements $P$.
\end{proof}

\printbibliography

\end{document}